\documentclass{eptcs}

\usepackage{breakurl}             % Not needed if you use pdflatex only.
\input{gaussian.sty}

\title{Generators and Relations for $U_n(\Z[\frac{1}{2},i])$}

\author{Xiaoning Bian and Peter Selinger
  \institute{Dalhousie University}}

\providecommand{\titlerunning}{Generators and Relations for $U_n(\Z[\frac{1}{2},i])$}
\providecommand{\authorrunning}{Xiaoning Bian and Peter Selinger}
\providecommand{\publicationstatus}{}

\begin{document}
\maketitle

\begin{abstract}
  Consider the universal gate set for quantum computing consisting of
  the gates $X$, $CX$, $CCX$, $\omega\da H$, and $S$. All of these gates
  have matrix entries in the ring $\Z[\frac{1}{2},i]$, the smallest
  subring of the complex numbers containing $\frac{1}{2}$ and
  $i$. Amy, Glaudell, and Ross proved the converse, i.e., any unitary
  matrix with entries in $\Z[\frac{1}{2},i]$ can be realized by a
  quantum circuit over the above gate set using at most one
  ancilla. In this paper, we give a finite presentation by generators
  and relations of $U_n(\Z[\frac{1}{2},i])$, the group of unitary
  $n\times n$-matrices with entries in $\Z[\frac{1}{2},i]$.
\end{abstract}

% ----------------------------------------------------------------------
\section{Introduction}

It is well-known that certain useful classes of quantum circuits
correspond to particular subrings of the complex numbers. For example,
the well-known Clifford+$T$ circuits, which are generated by the gates
\[ X =  \begin{bmatrix}
    0 & 1 \\
    1 & 0
  \end{bmatrix},
  \quad H =  \frac{1}{\sqrt 2}\begin{bmatrix}
    1 & 1 \\
    1 & -1
  \end{bmatrix},
  \quad \omega = e^{\frac{\pi}{4}i}=\frac{1+i}{\sqrt 2},
  \quad T = \begin{bmatrix}
    1 & 0 \\
    0 & \omega
  \end{bmatrix},
  \quad CX = \begin{bmatrix}
    1 & 0 & 0 & 0 \\
    0 & 1 & 0 & 0 \\
    0 & 0 & 0 & 1 \\
    0 & 0 & 1 & 0 
  \end{bmatrix},
\]
can represent exactly those unitary matrices whose matrix entries are
in the ring $\Z[\sfrac{1}{\sqrt{2}},i]$, the smallest subring of the
complex numbers containing $\frac{1}{\sqrt{2}}$ and $i$. The
left-to-right implication is obvious, since all of the above
generators are matrices with entries in this ring. The converse was
proved in {\cite{Giles-Selinger}}, where it was shown that every
unitary matrix over $\Z[\sfrac{1}{\sqrt{2}},i]$ can be represented as
a Clifford+$T$ circuit using at most one ancilla.

This result was extended by {\cite{Ross-Glaudell-Amy}} to several
other subrings of the complex numbers. In this paper, we are concerned
with the ring $\Z[\frac{1}{2},i]$, which corresponds to the class of
quantum circuits generated by the gates $X$, $CX$, $CCX$, $K$,
and $S$. Here, $X$ and $CX$ are as above, $CCX$ is the
doubly-controlled $X$-gate, also known as the Toffoli gate, and
\[
K = \omega\da H = \frac{1}{1+i}\begin{bmatrix}
    1 & 1 \\
    1 & -1
  \end{bmatrix}, \quad S = \begin{bmatrix}
    1 & 0 \\
    0 & i
  \end{bmatrix}.
\]
An important problem in quantum computing is the simplification of
quantum circuits. Often, the goal is to minimize the resources
required for quantum computing as much as possible, for example by
reducing the number of uses of certain costly gates such as the
$T$-gate. There are many successful strategies for simplifying quantum
circuits {\cite{debeaudrap2020fast, de_Beaudrap_2020, Nam_2018,
    Kissinger-Wetering, HC-2018, AM-2019, ZhangCheng19}}.  Some of
these strategies are based on {\em rewriting}, which means applying
algebraic laws, such as $H^2=I$, to replace subcircuits by equivalent
ones. In designing such rewriting techniques, it can be useful to have
a complete set of such algebraic laws, i.e., a set of equations by
which any circuit can in principle be rewritten into any equivalent
circuit. While a complete set of such equations does not in itself
guarantee the existence of good rewriting strategies, it can be a
useful tool for designing such strategies, a method that is called
``semantically guided rewriting''.

In this paper, we contribute to the algebraic theory of circuits by
giving a presentation of the group $U_n(\Z[\frac{1}{2},i])$, the group
of unitary $n\times n$-matrices with entries in $\Z[\frac{1}{2},i]$,
in terms of generators and relations. This follows earlier work by
Greylyn {\cite{Greylyn}}, who gave a presentation for
$U_4(\Z[\sfrac{1}{\sqrt{2}},i)$, and Li et
  al. {\cite{Li-Ross-Selinger}}, who gave a presentation for
  $U_n(\Z[\frac{1}{2}])$. The question of giving a presentation for
  $U_n(\Z[\sfrac{1}{\sqrt{2}},i])$ is still open for $n\geq 5$.

% ----------------------------------------------------------------------
\section{Notation and background}

The concepts and results of this section are similar to those used in
\cite{Giles-Selinger,Greylyn,Ross-Glaudell-Amy,Li-Ross-Selinger}.

% ----------------------------------------------------------------------
\subsection{Rings}

As usual, we write $\N=\s{0,1,\ldots}$ for the set of natural numbers.
We also write $\Z$, $\R$, and $\C$ to denote the rings of integers,
real numbers, and complex numbers, respectively. We write $z\da$
for the complex conjugate of a complex number. Recall that in any
ring, a {\em unit} is an invertible element.

Let $\Z[i]=\s{a+bi\mid a,b\in\Z}$ be the ring of Gaussian integers.
It is a Euclidean domain, and therefore, division with remainder,
greatest common divisors, divisibility, and the concept of a prime
(i.e. a non-unit that cannot be written as the product of two
non-units) all can be used in $\zi$.  Let $\gamma = 1+i$. It is
well-known that $\gamma$ is a Gaussian prime \cite{ant-book}. Note
that $\gamma$ is a divisor of $2$, and $\gamma^2$ and $2$ divide each
other. Using Euclidean division, it is further easy to check that
\[
\zi/(\gamma) = \{0,1\},\quad
\zi /(2) = \zi/(\gamma^2) = \{0,1,i, 1+i\},\quad
\text{and}\quad
\zi /(\gamma^3) = \{\pm 1,\pm i, 0, 1+i, 1- i, 2\}.
\]

\begin{definition}
  For $x \in \Z[i]$, if $x \equiv 0 \Mod{\gamma}$, we say $x$ is
  \emph{even}, otherwise \emph{odd}.
\end{definition}

For example, $2+3i$ and $5+2i$ are odd, and $2+4i$ and $1+3i$ are
even. In particular, in our context, by an even Gaussian integer we
mean one that is divisible by $\gamma$, not necessarily by $2$.
However if $x\in \Z$, then $x$ is even (odd) as an integer if and only
if $x$ is even (odd) as a Gaussian integer. Therefore, our definition
of parity extends the usual one on the integers. The following lemmas
are straightforward.

\begin{lemma} \label{lem:parity-dagger-norm}
  Let $\alpha =a+bi \in \zi$. Then $\alpha$ is odd if and only if
  $\norm{\alpha}^2=a^2+b^2$ is an odd integer.
\end{lemma}

\begin{lemma} \label{lem:modulo-1-3}
  If $\alpha \in \zi$ is odd, then $\alpha \equiv \pm 1, \pm i
  \Mod{\gamma^3}$. In other words, $\alpha \equiv i^e \Mod{\gamma^3}$
  for some $e\in\s{0,1,2,3}$. Moreover, $\alpha \equiv i^e
  \Mod{\gamma^2}$ for some $e\in\s{0,1}$.
\end{lemma}

Let $\D=\Z[\frac{1}{2}]=\s{\frac{a}{2^k}\mid a\in\Z, k\in\N}$ be the
ring of {\em dyadic fractions}, i.e., fractions whose denominator is a
power of 2. Consider $\D[i] = \s{r+si \mid r, s \in
  \D} = \Z[\frac{1}{2},i]$.  For every $t \in \di$, there exists some
natural number $k$ such that $\gamma^kt \in \zi$. This motivates the
following definition.

\begin{definition}
  Let $t \in \di$. A natural number $k \in \N$ is called a
  \emph{denominator exponent} for $t$ if $\gamma^{k} t \in \Z[i]$. The
  least such $k$ is called the \emph{least denominator exponent} of
  $t$, denoted by $\ldeg t$. Equivalently, the least denominator
  exponent of $t$ is the smallest $k\in \N$ such that $t$ can be
  written in the form $\frac{s}{\gamma^k}$, where $s \in \Z[i]$.

  More generally, we say that $k$ is a denominator exponent for a
  vector or matrix if it is a denominator exponent for all of its
  entries. The least denominator exponent for a vector or matrix is
  therefore the least $k$ that is a denominator exponent for all of
  its entries.
\end{definition}

Note that for $t\in \di$, if $k = \ldeg{t}>0$, then $\gamma^kt$ is
odd.

% ----------------------------------------------------------------------
\subsection{One- and two-level matrices}

Consider complex matrices of dimension $n\times n$. We number the rows
and columns of matrices starting from zero, i.e., the entries of an
$n\times n$-matrix are $a_{00},\ldots,a_{n-1,n-1}$. We define a
special class of matrices called one- and two-level matrices.

\begin{definition}
  Given $z\in\C$ and $j\in\s{0,\ldots,n-1}$, the {\em one-level
    matrix} $z_{[\jay]}$ is
  \[
  z_{[\jay]} =
  \kbordermatrix{
    & \cdots & \jay & \cdots\\
    \svdots & I & 0 & 0 \\
    \jay & 0 & z & 0 \\
    \svdots & 0 & 0 & I 
  },
  \]
  i.e., the matrix that is like the $n\times n$-identity matrix,
  except that the entry at position $(\jay,\jay)$ is $z$. Similarly,
  given a $2\times 2$-matrix $U= \begin{bmatrix}
    a& b \\
    c& d
  \end{bmatrix}$
  and $\jay, \kay \in \{0,1,...,n-1\}$ with $\jay < \kay$, the
  \emph{two-level matrix} $U_{[\jay,\kay]}$ is
  \[U_{[j,k]}=
  \kbordermatrix{
    &   \dots & j & \dots & k & \dots \\
    \svdots &   I & 0 & 0 & 0 & 0 \\
    j &   0 & a & 0 & b & 0 \\
    \svdots &   0 & 0 & I & 0 & 0 \\
    k &   0 & c & 0 & d & 0 \\
    \svdots &   0 & 0 & 0 & 0 & I
  }.
  \]
\end{definition}

Note that if $|z|=1$, then $z_{[\jay]}$ is unitary; similarly, if $U$
is unitary, then so is $U_{[\jay,\kay]}$. We say that
$U_{[\jay,\kay]}$ is a two-level matrix of {\em type} $U$, and
similarly, we say that $z_{[\jay]}$ is a one-level matrix of {\em
  type} $z$.

% ----------------------------------------------------------------------
\subsection{The exact synthesis algorithm}\label{ssec:synthesis}

Let $U_n(\D[i])$ be the group of unitary matrices with entries in
$\D[i]$. In this section, we will show that every element of
$U_n(\D[i])$ can be written as a product of one- and two-level
matrices of the form $\tli{j}{}, \tlx{j}{k}$, and $\tlk{j}{k}$, where
$i$ is the imaginary unit and the matrices $X$ and $K$ are the ones
that were defined in the introduction. In other words, we will show
that these matrices are {\em generators} for the group $U_n(\D[i])$.
We will do so by exhibiting a particular algorithm that inputs a
matrix $U\in U_n(\D[i])$ and outputs a corresponding word in the
generators. In quantum computing, such an algorithm is called an {\em
  exact synthesis algorithm} (as opposed to an approximate synthesis
algorithm, which only approximates a unitary matrix up to some given
$\varepsilon$). The algorithm in this section is adapted from
\cite{Giles-Selinger} and \cite{Ross-Glaudell-Amy}.

\begin{remark}
  A slight technical inconvenience arises because the matrix $K$ is
  not self-inverse. Therefore, in the following presentation we
  sometimes use $K$ as a generator and sometimes its inverse
  $K\da$. We have carefully chosen which one to use in each instance,
  to make the proofs in the later parts of the paper as simple as
  possible.  However, readers who wish to ignore the difference
  between $K$ and $K\da$ can safely do so, because in any case we have
  the relation $K\da=iK$, so whatever is generated by $K$ is also
  generated by $K\da$ and vice versa.
\end{remark}

We start with a number of lemmas. Since $\zi$ and $\di$ are subrings
of $\C$, the usual properties of complex numbers, complex vectors, and
complex matrices, such as complex conjugation, inner product, norm,
and conjugate transposition are naturally inherited.  For example, for
a vector $v \in \di^n$, the \emph{norm} of $v$ is $\|v\|= \sqrt
{v^{\dagger}v}$, where $v^{\dagger}$ is the conjugate transpose of
$v$. Note that $\|v\|^2 = v^{\dagger}v \in \D$, since $v^{\dagger}v \in
\D[i] \cap \R$. Similarly if $v\in \zi$, then $v^{\dagger}v \in \Z$.
A {\em unit vector} is a vector of norm 1.

If $v$ is a vector, the notation $v_j$ refers to its $j$th entry;
entries are numbered starting from $j=0$.

\begin{lemma} \label{lem:lde0}
  Let $v$ be a unit vector in $\Z[i]^n$. Then $v$ has exactly one
  non-zero entry, and that entry is a power of $i$.
\end{lemma}

\begin{proof}
  First note that for $\alpha = a+ bi \in \zi$, $\|\alpha\|^2=a^2 +
  b^2$ is a non-negative integer, and $\|\alpha\| = 0$ if and only if
  $\alpha=0$.  Let $v = [v_0, v_1,...,v_{n-1}]^T$. By assumption,
  $\sum_{\jay=0}^{n-1}\|v_\jay\|^2=1$. Since each $\|v_\jay\|^2$ is a
  non-negative integer, there is exact one non-zero $v_\jay$ such that
  $\|v_\jay\|=1$. Then it is easy to see that $v_\jay \in\s{\pm 1, \pm
    i}$.
\end{proof}

\begin{corollary} \label{cor:base-case}
  Let $v$ be a unit vector in $\Z[i]^n$, and let
  $p\in\s{0,\ldots,n-1}$. Then there exists a matrix $G$ that is a
  product of one- and two-level matrices of types $i$ and $X$, such
  that $Gv=e_p$, where $e_p$ is the $p$th standard basis vector.
\end{corollary}

\begin{proof}
  By Lemma~\ref{lem:lde0}, $v$ has a unique non-zero entry
  $v_m\in\s{\pm1,\pm i}$. Let $e\in\s{0,1,2,3}$ such that $i^ev_m =
  1$. Define
  \[
  G = \begin{cases}
    \tli{m}{e} & \text{if $m = p$,} \\
    \tlx{m}{p}\tli{m}{e} & \text{otherwise.}\\
  \end{cases}
  \]
  Then $Gv = e_p$ as desired.
\end{proof}

\begin{lemma} \label{lem:evenodd}
  Let $v$ be a unit vector in $\D[i]^n$ with least denominator
  exponent $k > 0$, and let $w=\gamma^k v \in \zi^n$. Then $w$ has an
  even number of odd entries.
\end{lemma}

\begin{proof}
  We have
  \[
  \sum_{\jay=0}^{n-1} \|w_\jay\|^2 = \norm{w} = |\gamma|^{2k}\norm{v}
  = (\gamma\gamma^{\dagger})^k\equiv 0 \Mod{\gamma}.
  \]
  Therefore, there are an even number of $\jay \in \{0,...,n-1\}$ such
  that $\|w_\jay\|^2$ is odd. It follows by
  Lemma~\ref{lem:parity-dagger-norm} that an even number of $w_\jay$
  are odd.
\end{proof}

For brevity, we sometimes write $\alpha(\gamma^{\ell})$ to denote any
member of the congruence class of $\alpha$ modulo $\gamma^{\ell}$.  In
other words, we write $\alpha(\gamma^{\ell})$ for any element of
$\Z[i]$ of the form $\alpha+\beta\gamma^{\ell}$, when the value of
$\beta$ is not important.

\begin{lemma}[Row operation] \label{lem:row-op}
  Let $v\in \D[i]^{n}$ be a vector with $\ldeg{v_\jay} = \ldeg{v_\ell}
  = k>0$. Then there exists an exponent $q \in\{0,1\}$ such that,
  if we let $v' = \tlkd{\jay}{\ell}\tli{\ell}{q} v$, then
  $\ldeg{v'_\jay} , \ldeg{v'_\ell} <k$. The remaining entries of $v$
  are unchanged.
\end{lemma}

\begin{proof}
  Let $w_\jay = \gamma^k v_\jay$ and $w_\ell = \gamma^k v_\ell$. Then
  both $w_\jay$ and $w_\ell$ are odd. By Lemma~\ref{lem:modulo-1-3},
  there exists $q \in \s{0,1}$ such that $w_\jay \equiv i^q
  w_\ell\Mod{\gamma^2}$. Then we have
  \[
  K\da\tli{1}{q} \begin{bmatrix}
    w_\jay \\
    w_\ell
  \end{bmatrix} =
  K\da\begin{bmatrix}
    w_\jay \\
    i^qw_\ell
  \end{bmatrix}
  = K\da \begin{bmatrix}
    w_\jay \\
    w_\jay + \alpha\gamma^2
  \end{bmatrix}
  = \frac{1}{\gamma\da}\begin{bmatrix}
    2w_j+\alpha\gamma^2 \\
    -\alpha\gamma^2
  \end{bmatrix}
  = \begin{bmatrix}
    (w_j+\alpha i)\gamma \\
    -\alpha i\gamma
  \end{bmatrix}
  = \begin{bmatrix}
    0(\gamma) \\
    0(\gamma)
  \end{bmatrix}.
  \]
  This means $\ldeg{K\da\tli{1}{q} \begin{bmatrix} v_\jay
      \\ v_\ell\end{bmatrix}}<k$. Clearly, performing
  $\tlkd{\jay}{\ell}\tli{\ell}{q}$ on $v$ has the same effect, and
  does not change any entries of $v$ besides $v_\jay$ and $v_\ell$,
  proving the lemma.
\end{proof}

\begin{lemma}[Column lemma]
  Let $v$ be a unit vector in $\D[i]^n$, and
  $p\in\{0,1,...,n-1\}$. Then there exists a matrix $G$ that is a
  product of one- and two-level matrices of types $X$, $K\da$, and $i$,
  such that $Gv=e_p$.
\end{lemma}

\begin{proof}
  By induction on $k$, the least denominator exponent of $v$, and
  nested induction on the number of entries of $v$ with least
  denominator exponent $k$. If $k=0$, then the result follows from
  Corollary~\ref{cor:base-case}. Otherwise, by
  Lemma~\ref{lem:evenodd}, $v$ has at least two entries with least
  denominator exponent $k$. Pick a pair of such
  entries. Lemma~\ref{lem:row-op} yields some one- and two level
  matrices that decrease this pair's least denominator
  exponent. The result then follows by the induction hypothesis.
\end{proof}

Since each column of a unitary matrix is a unit vector, the column
lemma naturally gives a way to ``fix'' every column of a unitary
matrix to the $j$th standard basis vector.

\begin{lemma}[Matrix decomposition] \label{lem:matrix-decomp}
  Let $U$ be a unitary $n \times n$-matrix with entries in
  $\D[i]$. Then there exists a matrix $G$ that is a product of one-
  and two-level matrices of types $X$, $K\da$, and $i$ such that $GU=I$.
\end{lemma}

\begin{proof}
  This is an easy consequence of the column lemma. Specifically, first
  use the column lemma to find a suitable $G_1$ such that the
  rightmost column of $G_1 U$ is $e_{n-1}$. Because $G_1 U$ is
  unitary, it is of the form
  \[
  \left[\begin{array}{c|c}
    U^{\prime} & 0\\
    \hline
    0 & 1 
  \end{array}\right].
  \]
Now recursively find one- and two-level matrices to reduce
$U^{\prime}$ to the identity matrix.
\end{proof}

\begin{corollary}\label{cor:generating}
  A matrix $U$ with entries in $\D[i]$ is unitary if and only if it
  can be written as a product of one- and two-level matrices of types
  $X$, $K$, and $i$.
\end{corollary}

\begin{proof}
  The right-to-left implication is obvious, because the relevant one-
  and two-level matrices are themselves unitary. The left-to-right
  implication follows from
  Lemma~\ref{lem:matrix-decomp}. Specifically, by
  Lemma~\ref{lem:matrix-decomp}, we can find a product $G$ of one- and
  two-level matrices of types $X$, $K\da$, and $i$ such that
  $GU=I$. Then $U=G^{-1}$. Since $X^{-1}=X$, $(K\da)^{-1}=K$, and
  $i^{-1}=i^3$, the inverse $G^{-1}$ can be expressed as a product of
  one- and two-level matrices of the required types.
\end{proof}

Since the proof of Lemma~\ref{lem:matrix-decomp} is constructive, it
yields an algorithm that inputs an element of $U_n(\D[i])$ and
outputs a sequence $G$ of one- and two-level matrices of types $X$,
$K\da$, and $i$ such that $GU=I$. We call this an ``exact synthesis
algorithm''. In principle, there are many different ways to achieve
this; for example, one step in the proof of the column lemma requires
us to ``pick a pair'' of entries, and the computed sequence of
generators will depend on which pair we pick.  For the rest of this
paper, it is important to fix, once and for all, a particular
deterministic method for making these choices. Therefore, we specify
one such deterministic procedure in Algorithm~\ref{algo}.

In the following, by the {\em pivot column} of a matrix $U$, we mean
the rightmost column where $U$ differs from the identity matrix.

\begin{algorithm}[Exact synthesis algorithm]\label{algo}~

  INPUT: A unitary $n \times n$-matrix $U$ with entries in $\D[i]$.

  OUTPUT: A sequence of one- and two-level matrices $G_{l}, \ldots,
  G_{1}$ of types $X$, $K\da$, and $i$ such that $G_{l} \cdots G_{1}
  U=I$.

  STATE: Let $M$ be a storage for a unitary $n \times n$-matrix, and
  let $\vec G$ be a storage for a sequence of one- and two-level
  matrices.
  \begin{itemize}
  \item [1.] Set $M$ to be $U$, and set $\vec G$ to be the empty
    sequence.
  \item [2.] If $M=I$, stop, and output $\vec G$.
  \item [3.] Let $p$ be the index of the pivot column of $M$, and let
    $k$ be the least denominator exponent of the pivot column of $M$.
    \begin{itemize}
    \item [(a)] If $k=0$, let $\vec S$ be obtained by applying
      Corollary~\ref{cor:base-case} to the pivot column and $p$.
      
    \item [(b)] If $k>0$, Lemma~\ref{lem:evenodd} dictates there are
      an even number of entries that have least denominator exponent
      $k$. Let the indices of the first two odd entries be
      $\jay<\ell$. Let $\vec S$ be obtained by applying
      Lemma~\ref{lem:row-op} to the pivot column and $\jay,\ell$.
    \end{itemize}
    
  \item [4.] Set $M$ to be $\vec S M$, and prepend $\vec S$ to $\vec
    G$. Go to step 2.
  \end{itemize}
\end{algorithm}

We refer to the $\vec S$ in step 3 as a {\em syllable}. We can also
regard $\vec G$ as a sequence of syllables, so we can talk about the
$n$th syllable in $\vec G$.
  
The fact that the algorithm is correct and terminating is just another
way of stating Lemma~\ref{lem:matrix-decomp}. However, we will provide
a more explicit proof of correctness and termination, as this will be
useful later in this paper.

\begin{definition}\label{def:level}
  Let $M \in U_{n} (\di)$. The {\em level} of $M$, denoted by
  $\level{M}$, is defined as follows: If $M=I$, then
  $\level{M}=(0,0,0)$.  Otherwise, $\level{M}$ is a triple $(p, k, m)$
  where:
  \begin{itemize}
\item $p$ is the greatest index such that $M e_{p} \neq e_{p}$,
  i.e., the index of the pivot column.
\item $k = \lde v$, where $v$ is the pivot column. 
\item $m$ is the number of odd entries in $\gamma^k v$.
  \end{itemize}
\end{definition}

Note that $p,k,m$ are natural numbers, so the set of all possible
levels is a subset of $\N^3$. We use the lexicographic order on
$\N^3$, defined by $(p,k,m)<(p',k',m')$ iff $p<p'$ or ($p=p'$ and
$k<k'$) or ($p=p'$ and $k=k'$ and $m<m'$. This makes $\N^3$ into a
well-ordered set.

\begin{theorem}\label{thm:algo-correctness}
  Given a unitary matrix $U \in U_n(\di)$, Algorithm~\ref{algo}
  outputs a finite sequence of one- and two-level matrices $ G_{l},
  \ldots, G_{1}$ such that $ G_{l} \cdots G_{1} U=I$.
\end{theorem}

\begin{proof}
  For correctness, it suffices to note that after the initialization
  in step 1, $M$ and $\vec G$ are only updated in step 4, and at each
  point in the algorithm, we have $M = \vec GU$. The algorithm only
  stops when $M=I$, at which point we therefore have $I=\vec GU$. The
  only thing that remains to be shown is that the algorithm
  terminates.

  We prove termination by well-ordered induction on the level of $M$.
  Specifically, we prove that after the algorithm reaches step 2, it
  always terminates. When $M=I$, this is trivially true. Otherwise,
  step 3 yields a syllable $\vec S$ such that $\level{\vec S M} <
  \level{M}$. Step 4 sets $M$ to $\vec S M$, and loops back to step
  2. At this point, since $M$ now has a strictly smaller level, the
  algorithm terminates by the induction hypothesis.
\end{proof}

% ----------------------------------------------------------------------
\subsection{Monoid presentations}

We recall some basic definitions and facts about generators and
relations for monoids. A {\em monoid} is a set $G$ together with an
associative multiplication operation and a unit 1. A subset $H\seq G$
is called a {\em submonoid} of $G$ if $1\in H$ and $H$ is closed under
multiplication. Given a subset $X\seq G$, let $\gen{X}$ be the
smallest submonoid of $G$ containing $X$. We say that $X$ is a set of
{\em generators} for $G$ if $G=\gen{X}$.

A {\em congruence} on $G$ is an equivalence relation that is
compatible with multiplication. In other words, a relation ${\sim}
\subseteq G \times G$ is a congruence if it obeys the following rules:
\begin{equation*} \label{cong-rel}
  \frac{}{w \sim w} \qquad
  \frac{w \sim v \quad v \sim u}{w \sim u} \qquad
  \frac{w \sim v}{v \sim w} \qquad
  \frac{w \sim w^{\prime} \quad v \sim v^{\prime}}{w v \sim w^{\prime} v^{\prime}}.
\end{equation*}

If $\sim$ is a congruence, then the quotient $G / {\sim}$ is a
well-defined monoid. If $f: G \to H$ is a monoid homomorphism, then
its kernel is the relation $\sim_{f}$ defined by $u \sim_{f} v$ if and
only if $f(u)=f(v)$. The kernel of every monoid homomorphism is a
congruence, and conversely, every congruence $\sim$ is the kernel of
some homomorphism, namely of the canonical quotient map $G\to
G/{\sim}$.

Given a set $X$, let $X^{*}$ be the set of finite sequences of
elements of $X$, which we also call {\em words} over the alphabet
$X$. Then $X^{*}$ is a monoid, where multiplication is given by
concatenation of words, and the unit is the empty word, which we also
write as $\eps$. By abuse of notation, we regard $X$ as a subset of
$X^{*}$, i.e., we identify each element $x \in X$ with the
corresponding one-letter word $x \in X^{*}$.  Given a set $X$, a {\em
  relation} is a pair $(w, v)$ of words, where $w, v \in X^{*}$. Given
a set $\Rr$ of relations for $X$, the set of consequences of
$\Rr$, written $\simr$, is the smallest
congruence relation on $X^{*}$ containing $\Rr$.

If $X$ is a set of generators for some monoid $G$, and $w=x_{1} \ldots
x_{n} \in X^{*}$ is a word, let us write $\sem{w}=x_{1} \cdots x_{n}$ for
the element of $G$ obtained by multiplying $x_{1}, \ldots, x_{n}$. In
other words, $\sem{-}: X^{*} \to G$ is the unique monoid homomorphisms
such that $\sem{x}=x$ for all $x \in X$. We say that a relation $(w, v)$
is {\em true} in $G$ if $\sem{w}=\sem{v}$. A set of relations $\Rr$ is {\em
  sound} for $G$ if all of its consequences are true in $G$, i.e., if
$w \simr v \Rightarrow\sem{w}=\sem{v}$.  A set of relations $\Rr$ is {\em
  complete} for $G$ if all true relations are consequences of $\Rr$,
i.e., if $\sem{w}=\sem{v} \Rightarrow w \simr v$.  We say that $(X, \Rr)$
is a {\em presentation} of $G$ (also called a presentation by
generators and relations) if $X$ is a set of generators for $G$ and
$\Rr$ is a sound and complete set of relations for $G$. In this case,
we have $G\iso X^{*} / {\simr}$, i.e., the monoid $G$ is
uniquely determined, up to isomorphism, by such a presentation.

\begin{remark} \label{rem:soundness}
  To check that a set of relations $\Rr$ is sound for $G$, it suffices
  to check that $\sem{w}=\sem{v}$ holds for all $(w, v) \in \Rr$.
\end{remark}

The monoids we consider in this paper are actually groups, i.e., all
of their elements are invertible. However, when we speak of
presentations, generators, relations, congruences, etc., we still mean
this ``in the monoid sense''. In particular, we will include enough
generators and relations to imply the invertibility of all elements.
This differs from presentations ``in the group sense'', where inverses
would exist automatically, and relations such as $g^{-1} g=\eps$ would
be unnecessary. The monoid approach is advantageous for our purposes,
because it is much more straightforward to do calculations in free
monoids (where the elements are words) than in free groups (where the
elements are group-theoretic expressions modulo an equivalence
relation).

% ----------------------------------------------------------------------
\section{A presentation of $U_n(\di)$}

% ----------------------------------------------------------------------
\subsection{Statement of the main theorem}

Let $\Gg$ be the set of all one- and two-level matrices of types $X$,
$K$, and $i$. We refer to the elements of $\Gg$ as the {\em
  generators}. Specifically, they are:
\begin{itemize}
\item $\tlx{\jay}{\ell}$, for $0\leq\jay<\ell<n$;
\item $\tlk{\jay}{\ell}$, for $0\leq\jay<\ell<n$; and
\item $\tli{\jay}{}$, for $0\leq\jay<n$.
\end{itemize}
We already saw in Corollary~\ref{cor:generating} that $\Gg$ indeed
generates $U_n(\di)$. The purpose of this paper is to prove the
following theorem:

\begin{theorem}[Main theorem]\label{thm:main}
  Let $\Rr$ be the set of relations shown in Figure~\ref{rels}. Then
  $(\Gg,\Rr)$ is a presentation of $U_n(\di)$.  In other words, the
  relations in Figure~\ref{rels} are sound and complete for
  $U_n(\di)$.
\end{theorem}

In Figure~\ref{rels}, relations {\eqref{eq:1}}--{\eqref{eq:3}} give
the order of the generators. They also ensure that all of the
generators (and therefore all words of generators) are
invertible. Relations {\eqref{eq:4}}--{\eqref{eq:9}} state that
generators with disjoint indices commute. Relations
{\eqref{eq:10}}--{\eqref{eq:12}} state that $\tlx{j}{k}$ can be used
to swap indices $j$ and $k$ in any other generator. Finally, relations
{\eqref{eq:13}}--{\eqref{eq:17}} state additional properties of the
generators. We note that the relations in Figure~\ref{rels} are not
minimal; for example, {\eqref{eq:1}} and {\eqref{eq:16}} imply
{\eqref{eq:3}}. However, we think that they are a ``nice'' set of
relations.

\begin{remark}
  To prove the soundness part of Theorem~\ref{thm:main}, by
  Remark~\ref{rem:soundness}, it suffices to check that each relation
  in $\Rr$ is true in $U_n(\di)$. This can be verified by calculation.
  The remainder of this paper is devoted to proving completeness.
\end{remark}

For convenience, we say that two words $w,u$ are {\em relationally
  equivalent} if $w\simr u$ and {\em semantically equivalent} if
$\sem{w}=\sem{u}$. Completeness is thus the statement that all
semantically equivalent words are relationally equivalent, and
soundness is the converse.

% ......................................................................
\begin{figure}
  \begin{minipage}{0.37\textwidth}
    \begin{align}
      \tli{j}{4} &\,\sim\, \eps \label{eq:1}\\
      \tlx{j}{k}^2 &\,\sim\, \eps \label{eq:2}\\
      \tlk{j}{k}^8 &\,\sim\, \eps \label{eq:3}\\
      \nonumber \\ 
      \tli{j}{}\tli{k}{} & \,\sim\, \tli{k}{}\tli{j}{} \label{eq:4}\\
      \tli{j}{}\tlx{k}{\ell} &\,\sim\,  \tlx{k}{\ell}\tli{j}{}  \label{eq:5}\\
      \tli{j}{}\tlk{k}{\ell} &\,\sim\,  \tlk{k}{\ell}\tli{j}{}  \label{eq:6}\\
      \tlx{j}{k}\tlx{\ell}{m} &\,\sim\, \tlx{\ell}{m}\tlx{j}{k}  \label{eq:7}\\
      \tlx{j}{k}\tlk{\ell}{m} &\,\sim\, \tlk{\ell}{m}\tlx{j}{k}  \label{eq:8}\\
      \tlk{j}{k}\tlk{\ell}{m} &\,\sim\, \tlk{\ell}{m}\tlk{j}{k}  \label{eq:9}
    \end{align}
  \end{minipage}
  \begin{minipage}{0.58\textwidth}
    \begin{align}
      \tli{k}{}\tlx{j}{k} &\,\sim\, \tlx{j}{k}\tli{j}{} \label{eq:10}\\
      \tlx{k}{\ell}\tlx{j}{k} &\,\sim\, \tlx{j}{k}\tlx{j}{\ell} \label{eq:11}\\
      \tlx{\jay}{\ell}\tlx{k}{\ell} &\,\sim\, \tlx{k}{\ell}\tlx{j}{k} \label{eq:11'}\\
      \tlk{k}{\ell}\tlx{j}{k} &\,\sim\, \tlx{j}{k}\tlk{j}{\ell} \label{eq:12}\\
      \tlk{\jay}{\ell}\tlx{k}{\ell} &\,\sim\, \tlx{k}{\ell}\tlk{j}{k} \label{eq:12'}\\
      \nonumber \\ 
    \tlk{j}{k}\tli{k}{2} &\,\sim\, \tlx{j}{k}\tlk{j}{k} \label{eq:13}\\
    \tlk{j}{k}\tli{k}{3}  &\,\sim\, \tli{k}{}\tlk{j}{k}\tli{k}{}\tlk{j}{k} \label{eq:14}\\
    \tlk{j}{k}\tli{j}{}\tli{k}{} &\,\sim\, \tli{j}{}\tli{k}{}\tlk{j}{k} \label{eq:15}\\
    \tlk{j}{k}^2\tli{j}{}\tli{k}{} &\,\sim\, \eps  \label{eq:16}\\
    \tlk{j}{k}\tlk{\ell}{m} \tlk{j}{\ell}\tlk{k}{m}  &\,\sim\,  \tlk{j}{\ell}\tlk{k}{m} \tlk{j}{k}\tlk{\ell}{m}\label{eq:17}
    \end{align}
  \end{minipage}
  \caption{A sound and complete set of relations for $U_n(\di)$. In
    each relation, the indices are assumed to be distinct; moreover,
    whenever a generator $\tlx{a}{b}$ or $\tlk{a}{b}$ is mentioned, we
    assume $a<b$.}\label{rels}
  \hfill\rule{0.95\textwidth}{0.1mm}\hfill\hbox{}
\end{figure}
% ......................................................................

% ----------------------------------------------------------------------
\subsection{The Cayley graph}

In order to conceptualize the proof of Theorem~\ref{thm:main}, it is
useful to define a particular kind of infinite directed graph, which we
call the {\em Cayley graph} of $U_n(\di)$. The vertices of the Cayley
graph are the elements of $U_n(\di)$, i.e., unitary matrices with
entries in $\di$. The identity matrix $I$ plays a special role and we
call it the {\em root} of the graph. We often use letters such as
$s,r,t,q$ to denote vertices of the Cayley graph. Our Cayley graph
contains two different kinds of edges:
\begin{itemize}
\item {\em Simple edges.} Simple edges are labelled by generators
  $G\in\Gg$ and denoted by single arrows $\xrightarrow{G}$. There is
  an edge $s\xrightarrow{G}t$ if and only if $Gs = t$. For the latter
  equation to make sense, keep in mind that $G$, $s$, and $t$ are all
  elements of $U_n(\di)$.
\item {\em Normal edges.} Recall from Section~\ref{ssec:synthesis}
  that each word $\vec S$ that is produced in step 3 of
  Algorithm~\ref{algo} is called a {\em syllable}. From now on, we
  often write $N$ to denote a single such syllable (which may,
  however, be a word consisting of more than one generator). Normal
  edges are labelled by syllables, and are denoted by double arrows
  $\xRightarrow{N}$. For every non-root vertex $s$ of the Cayley
  graph, there is a unique normal edge originating at $s$, and it is
  given by $s\xRightarrow{N}t$, where $N$ is the syllable that
  Algorithm~\ref{algo} produces when $M=s$, and $t=\vec Ss$.
\end{itemize}

Note that as a consequence of these definitions, the Cayley graph has
a tree structure with respect to the normal edges. Specifically, for
every vertex $s$, there exists a unique path of normal edges starting
at $s$. Also note that if $s\xRightarrow{N} t$, then
$\level{t}<\level{s}$, so the path of normal edges starting at any
given vertex $s$ is necessarily finite. Since the root $I$ is the only
vertex with no outgoing normal edge, every such path therefore
necessarily ends at $I$, i.e., it is of the form
$s=s_0\xRightarrow{N_1}s_1\xRightarrow{N_2}\ldots\xRightarrow{N_m}s_m=I$,
for $m\geq 0$.

\begin{definition}
  Given any matrix $U\in U_n(\di)$, its {\em normal word} is the word
  $w\in\Gg^*$ defined as follows: Let
  $U=s_0\xRightarrow{N_1}s_1\xRightarrow{N_2}\ldots\xRightarrow{N_m}s_m=I$
  be the unique path of normal edges from $U$ to $I$ in the Cayley
  graph. Then we define $w$ to be the concatenation $N_m\cdots N_1$,
  where each occurrence of $K\da$ is replaced by $K^7$. Equivalently,
  $w$ is the word output by Algorithm~\ref{algo} on input $U$. Note
  that, by definition of the Cayley graph or equivalently by
  Theorem~\ref{thm:algo-correctness}, we have $\sem{w}U = I$, or in
  other words $\sem{w} = U^{-1}$. We write $w=\normalword(U)$ when $w$
  is the normal word of $U$.

  Moreover, given any word $u\in\Gg^*$, not necessarily normal, define
  its {\em normal form} to be $\normalword(\sem{u}^{-1})$.  Note that
  by definition, every word $u$ has a unique normal form, and if $w$
  is the normal form of $u$, then $\sem{w}=\sem{u}$. Moreover, if
  $\sem{u}=\sem{u'}$ then $u$ and $u'$ have the same normal form.
  We write $\nf(u)$ for the normal form of a word $u$.
\end{definition}

Given a path $\vec G=s_0\xrightarrow{G_1}s_1\xrightarrow{G_2}s_2\ldots
s_{n-1}\xrightarrow{G_n}s_n$ of simple edges in the Cayley graph, we
define the {\em level} of the path by $\level{\vec G} =
\max\s{\level{s_i} \mid i=0,\ldots,n}$. Here, the level of a vertex is
of course as defined in Definition~\ref{def:level}.

% ----------------------------------------------------------------------
\subsection{Basic generators}

To simplify the proof of completeness, it will sometimes be useful to
consider a smaller set of generators for $U_n(\di)$, which we call the
{\em basic generators}. They are the following:
\begin{itemize}
\item $\tlx{\jay}{\jay+1}$, for $0\leq\jay<n-1$;
\item $\tlk{0}{1}$; and
\item $\tli{0}{}$.
\end{itemize}
We also call the corresponding edges of the Cayley graph {\em basic
  edges}.

\begin{lemma}\label{lem:basic-generators}
  Each generator $G\in\Gg$ can be written as a product of basic
  generators by means of repeated applications of the following
  relations:
  \begin{align*}
    \tli{\jay}{} &\simr \tlx{0}{\jay}\tli{0}{}\tlx{0}{\jay} && \mbox{for $j>1$,}\\
    \tlk{\jay}{\ell} &\simr \tlx{0}{\jay}\tlk{0}{\ell}\tlx{0}{\jay} && \mbox{for $\jay>0$,}\\
    \tlk{0}{\ell} &\simr \tlx{1}{\ell}\tlk{0}{1}\tlx{1}{\ell} && \mbox{for $\ell>1$,}\\
    \tlx{\jay}{\ell} &\simr \tlx{\jay}{\jay+1}\tlx{\jay+1}{\ell}\tlx{\jay}{\jay+1} && \mbox{for $\ell>\jay+1$.}
  \end{align*}
\end{lemma}

\begin{proof}
  By case distinction. For example, by repeated application of these
  relations, we have:
  \[
  \def\a{\tlx{0}{1}\tlx{1}{2}\tlx{0}{1}}
  \def\b{\tlx{1}{2}\tlx{2}{3}\tlx{3}{4}\tlx{2}{3}\tlx{1}{2}}
  \tlk{2}{4} \simr \a\b\tlk{0}{1}\b\a.\qedhere
  \]
\end{proof}

Moreover, we will use the following fact: if $s\xrightarrow{G}r$ is a
simple edge of the Cayley graph and
$s\xrightarrow{G_1}s_1\xrightarrow{G_2}\ldots\xrightarrow{G_m}r$ is
the corresponding sequence of basic edges obtained from the above
relations, then $\level{s_1},\ldots,\level{s_{m-1}}\leq
\max(\level{s},\level{r})$. In other words, the conversion to basic
generators does not increase the level.

% ----------------------------------------------------------------------
\subsection{Reduction of completeness to the Main Lemma}

We now outline our strategy for proving completeness. For pedagogical
reasons, the following lemmas are stated in the opposite order in
which they are proved, i.e., each lemma implies the one before it.
Completeness is a direct consequence of the following lemma.

\begin{lemma}\label{lem:maincoro}
  Every word is relationally equivalent to its normal form.
\end{lemma}

To see why the lemma implies completeness, let $v,u$ be any two words
such that $\sem{v}=\sem{u}$. Then $v$ and $u$ have the same normal
form, say $w$. By the lemma, $v\simr w\simr u$, from which
completeness follows by transitivity.

We will prove Lemma~\ref{lem:maincoro} by induction on the length of
the word $u$. Since by Lemma~\ref{lem:basic-generators}, each
generator is relationally equivalent to a word of basic generators, we
can assume without loss of generality that $u$ consists of basic
generators. For the induction step, we use the following lemma:

\begin{lemma}\label{lem:maincoro-step}
  Consider any basic edge $s\xrightarrow{G}r$ of the Cayley graph.
  Let $w=\normalword(s)$ and $u=\normalword(r)$. Then
  $w\simr uG$. Or in pictures, the following diagram commutes
  relationally:
  
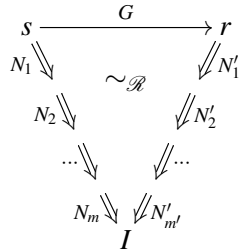
\[
  \xymatrix@C=0.5ex@R=3ex{
    s \ar[rrrrrrrr]^<>(.5){G}\ar@{=>}[dr]_<>(.4){N_1} &&&&&&&& r\ar@{=>}[dl]^<>(.4){N'_1}\\
    &{}\ar@{=>}[dr]_<>(.4){N_2}\ar@{}[rrrrrr]|<>(.5){\displaystyle\simr}&&&&&&{}\ar@{=>}[dl]^<>(.4){N'_2}\\
    &&{}\ar@{=>}[dr]_<>(.4){\cdots}&&&&{}\ar@{=>}[dl]^<>(.4){\cdots}\\
    &&&{}\ar@{=>}[dr]_<>(.4){N_m}&&{}\ar@{=>}[dl]^<>(.4){N'_{m'}}\\
    &&&&I \\
  }
  \]
\end{lemma}

To see why Lemma~\ref{lem:maincoro-step} implies
Lemma~\ref{lem:maincoro}, consider any word $u$ composed of basic
generators.  If $u=\eps$, then $u$ is already normal so there is
nothing to show. Otherwise, $u=u'G$ for some generator $G$. Let
$s=\sem{u}^{-1}$ and $r=\sem{u'}^{-1}$; then $s\xrightarrow{G}r$ is a
basic edge of the Cayley graph. By Lemma~\ref{lem:maincoro-step}, we
have $\normalword(r)\,G\simr\normalword(s)$. Also, by the induction
hypothesis, since $u'$ is a shorter word than $u$, we have $u'\simr
\nf(u')$. Then the claim follows because
$u=u'G\simr\nf(u')\,G=\normalword(r)\,G \simr
\normalword(s)=\nf(u)$. The following diagram illustrates the proof in
case $u=G_4G_3G_2G_1$ is a word of length 4.
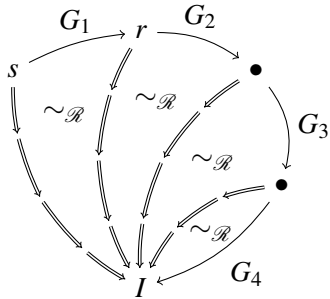
\[
\begin{tikzpicture}[scale=1.7]
  \tikzset{inner node/.style={inner sep=0,outer sep=0.5}}
  \node(s0) at (-1,1.7) {$s$};
  \node(s1) at (0,2) {$r$};
  \node(s2) at (0.9,1.7) {$\bullet$};
  \node(s3) at (1.1,0.8) {$\bullet$};
  \node(I) at (0,0) {$I$};
  \draw [->] (s0) edge[bend left=10] node[above] {$G_1$} (s1);
  \draw [->] (s1) edge[bend left=17] node[above] {$G_2$} (s2);
  \draw [->] (s2) edge[bend left=25] node[right] {$G_3$} (s3);
  \draw [->] (s3) edge[bend left=17] node[below right] {$G_4$} (I);
  \path [inner node] (s0) edge[bend right,draw=none]
  node(s01) [pos=0.2] {}
  node(s02) [pos=0.5] {}
  node(s03) [pos=0.8] {}
  (I);
  \draw [] (s0) edge[-implies,double,bend right=5] (s01);
  \draw [->] (s01) edge[-implies,double,bend right=5] (s02);
  \draw [->] (s02) edge[-implies,double,bend right=5] (s03);
  \draw [->] (s03) edge[-implies,double,bend right=5] (I);
  \path [inner node] (s1) edge[bend right,draw=none]
  node(s11) [pos=0.2] {}
  node(s12) [pos=0.5] {}
  node(s13) [pos=0.8] {}
  (I);
  \draw [] (s1) edge[-implies,double,bend right=5] (s11);
  \draw [->] (s11) edge[-implies,double,bend right=5] (s12);
  \draw [->] (s12) edge[-implies,double,bend right=5] (s13);
  \draw [->] (s13) edge[-implies,double,bend right=5] (I);
  \path [inner node] (s2) edge[bend right,draw=none]
  node(s21) [pos=0.2] {}
  node(s22) [pos=0.5] {}
  node(s23) [pos=0.8] {}
  (I);
  \draw [] (s2) edge[-implies,double,bend right=5] (s21);
  \draw [->] (s21) edge[-implies,double,bend right=5] (s22);
  \draw [->] (s22) edge[-implies,double,bend right=5] (s23);
  \draw [->] (s23) edge[-implies,double,bend right=5] (I);
  \path [inner node] (s3) edge[bend right,draw=none]
  node(s31) [pos=0.3] {}
  node(s32) [pos=0.7] {}
  (I);
  \draw [] (s3) edge[-implies,double,bend right=5] (s31);
  \draw [->] (s31) edge[-implies,double,bend right=5] (s32);
  \draw [->] (s32) edge[-implies,double,bend right=5] (I);
  \path (s01) --node[pos=0.5] {$\simr$} (s11);
  \path (s11) --node[pos=0.5] {$\simr$} (s21);
  \path (s21) --node[pos=0.6] {$\simr$} (s31);
  \path (s3) --node[pos=0.45,xshift=-2] {$\simr$} (I);
\end{tikzpicture}
\]

So now, what is left to do is to prove Lemma~\ref{lem:maincoro-step}.
We will prove this by induction on the level of $s$. The induction
step uses the following lemma.

\begin{lemma}[Main Lemma]\label{lem:main}
  Assume $s\xrightarrow{G}r$ is a basic edge of the Cayley graph, and
  $s\xRightarrow{N}t$ is a normal edge. Then there exists a sequence
  of normal edges $r\xRightarrow{\vec N'}q$ and a sequence of basic
  edges $t\xrightarrow{\vec G'}q$ such that $\vec N'G\simr \vec G'N$
  and $\level{\vec G'}<\level{s}$. Here is the relation $\vec N'G\simr
  \vec G'N$ as a diagram:
  \begin{equation}\label{eqn:lem:main}
    \vcenter{
      \xymatrix{
        s \ar[r]^{G}\ar@{=>}[d]_{N}\ar@{}[dr]|<>(.5){\displaystyle\simr} &
        r\ar@{=>}[d]^{\vec N'}\\
        t\ar@{->}[r]^{\vec G'} &
        q.
      }
    }
  \end{equation}
\end{lemma}

The proof that Lemma~\ref{lem:main} implies
Lemma~\ref{lem:maincoro-step} proceeds by induction on the level of
$s$. The base case arises when $s=I$. In that case, an easy case
distinction shows that the claim holds for all generators $G$.  For
the induction step, we have $s\neq I$, so there exists a unique normal
edge $s\xRightarrow{N}t$. By Lemma~\ref{lem:main}, there exists a
sequence of normal edges $r\xRightarrow{\vec N'}q$ and a sequence of
basic edges $t\xrightarrow{\vec G'}q$ such that
{\eqref{eqn:lem:main}} holds and $\level{\vec G'}<\level{s}$. Since
the level of each vertex occurring in the path $t\xrightarrow{\vec
  G'}q$ is strictly less than the level of $s$, by the induction
hypothesis, the claim of Lemma~\ref{lem:maincoro-step} is already true
for each edge in this path. Then the lemma follows by the following
diagram; note that each part commutes relationally and therefore so
does the whole diagram.
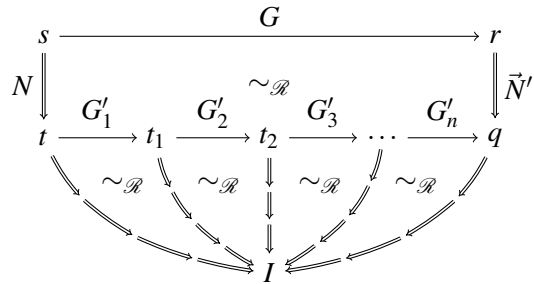
\[
\begin{tikzpicture}[xscale=1.5,yscale=0.9]
  \tikzset{inner node/.style={inner sep=0,outer sep=0.5}}
  \node(s) at (-2,3.5) {$s$};
  \node(r) at (2,3.5) {$r$};
  \node(t) at (-2,2) {$t$};
  \node(t1) at (-1,2) {$t_1$};
  \node(t2) at (0,2) {$t_2$};
  \node(t3) at (1,2) {$\ldots$};
  \node(q) at (2,2) {$q$};
  \node(I) at (0,0) {$I$};
  \draw[->] (s) -- node[above]{$G$} (r);
  \draw[-implies,double] (s) -- node[left]{$N$} (t);
  \draw[-implies,double] (r) -- node[right]{$\vec N'$} (q);
  \draw[->] (t) -- node[above]{$G'_1$} (t1);
  \draw[->] (t1) -- node[above]{$G'_2$} (t2);
  \draw[->] (t2) -- node[above]{$G'_3$} (t3);
  \draw[->] (t3) -- node[above]{$G'_n$} (q);
  \path [inner node] (t) edge[bend right,draw=none]
  node(t01) [pos=0.25] {}
  node(t02) [pos=0.55] {}
  node(t03) [pos=0.8] {}
  (I);
  \draw [] (t) edge[-implies,double,bend right=5] (t01);
  \draw [->] (t01) edge[-implies,double,bend right=5] (t02);
  \draw [->] (t02) edge[-implies,double,bend right=5] (t03);
  \draw [->] (t03) edge[-implies,double,bend right=5] (I);
  \path [inner node] (t1) edge[bend right=20,draw=none]
  node(t11) [pos=0.25] {}
  node(t12) [pos=0.55] {}
  node(t13) [pos=0.8] {}
  (I);
  \draw [] (t1) edge[-implies,double,bend right=5] (t11);
  \draw [->] (t11) edge[-implies,double,bend right=5] (t12);
  \draw [->] (t12) edge[-implies,double,bend right=5] (t13);
  \draw [->] (t13) edge[-implies,double,bend right=5] (I);
  \path [inner node] (t2) edge[draw=none]
  node(t21) [pos=0.33] {}
  node(t23) [pos=0.67] {}
  (I);
  \draw [] (t2) edge[-implies,double] (t21);
  \draw [->] (t21) edge[-implies,double] (t23);
  \draw [->] (t23) edge[-implies,double] (I);
  \path [inner node] (t3) edge[bend left=20,draw=none]
  node(t31) [pos=0.25] {}
  node(t32) [pos=0.55] {}
  node(t33) [pos=0.8] {}
  (I);
  \draw [] (t3) edge[-implies,double,bend left=5] (t31);
  \draw [->] (t31) edge[-implies,double,bend left=5] (t32);
  \draw [->] (t32) edge[-implies,double,bend left=5] (t33);
  \draw [->] (t33) edge[-implies,double,bend left=5] (I);
  \path [inner node] (q) edge[bend left,draw=none]
  node(t41) [pos=0.25] {}
  node(t42) [pos=0.55] {}
  node(t43) [pos=0.8] {}
  (I);
  \draw [] (q) edge[-implies,double,bend left=5] (t41);
  \draw [->] (t41) edge[-implies,double,bend left=5] (t42);
  \draw [->] (t42) edge[-implies,double,bend left=5] (t43);
  \draw [->] (t43) edge[-implies,double,bend left=5] (I);
  \node at (0,2.75) {$\simr$};
  \node at (-1.3,1.3) {$\simr$};
  \node at (-0.45,1.3) {$\simr$};
  \node at (0.45,1.3) {$\simr$};
  \node at (1.3,1.3) {$\simr$};
\end{tikzpicture}
\]
Given the above sequence of lemmas, to finish the completeness proof,
all that is now left to do is to prove Lemma~\ref{lem:main}. Since the
proof is a rather long and tedious case distinction, it can be found
in the Appendix.

% ----------------------------------------------------------------------
\section{Conclusion and future work}

We have given a presentation by generators and relations of the group
of unitary $n\times n$-matrices with entries in the ring $\D[i] =
\Z[\frac{1}{2},i]$. This matrix group has some applications in quantum
computing because it arises as the group of unitary operations that
are exactly representable by a certain gate set that is a subset of
the Clifford+$T$ circuits; namely, it is generated by the $X$,
controlled-$X$, Toffoli and $S$ gates along with a modified Hadamard
gate $K=\omega^{\dagger}H$. We have described this group in terms of
generators that are 1- and 2-level matrices. We do not yet have a
complete set of relations for the quantum circuits generated by the
above gates, and doing so is a non-trivial problem because quantum
circuits are described in terms of the tensor operation $\otimes$ and
not the direct sum of vector spaces.

While other subgroups of the Clifford+$T$ group have been described by
generators and relations {\cite{Li-Ross-Selinger}}, this has not yet
been done for the Clifford+$T$ group itself, except in the case of
matrices of size $4\times 4$ {\cite{Greylyn}}. Doing so would require
extending our results from the ring $\D[i]$ to the ring $\D[\omega]$.
This problem turns out to be harder than one would expect, because as
the complexity of the ring increases, it becomes more and more
difficult to complete all the cases of the Main Lemma. This is left
for future work.

%----------------------------------------------------------------------
\bibliographystyle{eptcs}
\bibliography{gaussian}

\begin{thebibliography}{10}
\providecommand{\bibitemdeclare}[2]{}
\providecommand{\surnamestart}{}
\providecommand{\surnameend}{}
\providecommand{\urlprefix}{Available at }
\providecommand{\url}[1]{\texttt{#1}}
\providecommand{\href}[2]{\texttt{#2}}
\providecommand{\urlalt}[2]{\href{#1}{#2}}
\providecommand{\doi}[1]{doi:\urlalt{http://dx.doi.org/#1}{#1}}
\providecommand{\bibinfo}[2]{#2}

\bibitemdeclare{article}{Ross-Glaudell-Amy}
\bibitem{Ross-Glaudell-Amy}
\bibinfo{author}{Matthew \surnamestart Amy\surnameend},
  \bibinfo{author}{Andrew~N. \surnamestart Glaudell\surnameend} \&
  \bibinfo{author}{Neil~J. \surnamestart Ross\surnameend}
  (\bibinfo{year}{2020}): \emph{\bibinfo{title}{Number-theoretic
  characterizations of some restricted {Clifford+$T$} circuits}}.
\newblock {\sl \bibinfo{journal}{Quantum}} \bibinfo{volume}{4}, p.
  \bibinfo{pages}{252}, \doi{10.22331/q-2020-04-06-252}.
\newblock \bibinfo{note}{Also available from \arxiv{1908.06076}}.

\bibitemdeclare{article}{AM-2019}
\bibitem{AM-2019}
\bibinfo{author}{Matthew \surnamestart {Amy}\surnameend} \&
  \bibinfo{author}{Michele \surnamestart {Mosca}\surnameend}
  (\bibinfo{year}{2019}): \emph{\bibinfo{title}{{$T$}-count optimization and
  {Reed}-{Muller} codes}}.
\newblock {\sl \bibinfo{journal}{IEEE Transactions on Information Theory}}
  \bibinfo{volume}{65}(\bibinfo{number}{8}), pp. \bibinfo{pages}{4771--4784},
  \doi{10.1109/TIT.2019.2906374}.
\newblock \bibinfo{note}{{Also available from \arxiv{1601.07363}}}.

\bibitemdeclare{inproceedings}{debeaudrap2020fast}
\bibitem{debeaudrap2020fast}
\bibinfo{author}{Niel \surnamestart de~Beaudrap\surnameend},
  \bibinfo{author}{Xiaoning \surnamestart Bian\surnameend} \&
  \bibinfo{author}{Quanlong \surnamestart Wang\surnameend}
  (\bibinfo{year}{2020}): \emph{\bibinfo{title}{{Fast and effective techniques
  for {T}-count reduction via spider nest identities}}}.
\newblock In \bibinfo{editor}{Steven~T. \surnamestart Flammia\surnameend},
  editor: {\sl \bibinfo{booktitle}{15th Conference on the Theory of Quantum
  Computation, Communication and Cryptography (TQC 2020)}}, {\sl
  \bibinfo{series}{Leibniz International Proceedings in Informatics (LIPIcs)}}
  \bibinfo{volume}{158}, \bibinfo{publisher}{Schloss Dagstuhl--Leibniz-Zentrum
  f{\"u}r Informatik}, \bibinfo{address}{Dagstuhl, Germany}, pp.
  \bibinfo{pages}{11:1--11:23}, \doi{10.4230/LIPIcs.TQC.2020.11}.
\newblock \bibinfo{note}{Also available from \arxiv{2004.05164}}.

\bibitemdeclare{inproceedings}{de_Beaudrap_2020}
\bibitem{de_Beaudrap_2020}
\bibinfo{author}{Niel \surnamestart de~Beaudrap\surnameend},
  \bibinfo{author}{Xiaoning \surnamestart Bian\surnameend} \&
  \bibinfo{author}{Quanlong \surnamestart Wang\surnameend}
  (\bibinfo{year}{2020}): \emph{\bibinfo{title}{Techniques to reduce
  $\pi/4$-parity-phase circuits, motivated by the {ZX} calculus}}.
\newblock In: {\sl \bibinfo{booktitle}{Proceedings of the 16th International
  Conference on Quantum Physics and Logic, QPL 2019}}, {\sl
  \bibinfo{series}{Electronic Proceedings in Theoretical Computer Science}}
  \bibinfo{volume}{318}, p. \bibinfo{pages}{131–149},
  \doi{10.4204/eptcs.318.9}.
\newblock \bibinfo{note}{Also available from \arxiv{1911.09039}}.

\bibitemdeclare{article}{Giles-Selinger}
\bibitem{Giles-Selinger}
\bibinfo{author}{Brett \surnamestart Giles\surnameend} \&
  \bibinfo{author}{Peter \surnamestart Selinger\surnameend}
  (\bibinfo{year}{2013}): \emph{\bibinfo{title}{Exact synthesis of multiqubit
  {Clifford+$T$} circuits}}.
\newblock {\sl \bibinfo{journal}{Physical Review A}} \bibinfo{volume}{87}, p.
  \bibinfo{pages}{032332 (7 pages)}, \doi{10.1103/PhysRevA.87.032332}.
\newblock \bibinfo{note}{Also available from \arxiv{1212.0506}}.

\bibitemdeclare{phdthesis}{Greylyn}
\bibitem{Greylyn}
\bibinfo{author}{Seth E.~M. \surnamestart Greylyn\surnameend}
  (\bibinfo{year}{2014}): \emph{\bibinfo{title}{Generators and relations for
  the group {$U_4(\mathbb{Z}[\frac{1}{\sqrt{2}},i])$}}}.
\newblock \bibinfo{type}{{M.Sc.\@ thesis}}, \bibinfo{school}{Dalhousie
  University}.
\newblock \bibinfo{note}{Available from \arxiv{1408.6204}}.

\bibitemdeclare{article}{HC-2018}
\bibitem{HC-2018}
\bibinfo{author}{Luke~E. \surnamestart Heyfron\surnameend} \&
  \bibinfo{author}{Earl~T. \surnamestart Campbell\surnameend}
  (\bibinfo{year}{2018}): \emph{\bibinfo{title}{An efficient quantum compiler
  that reduces {$T$} count}}.
\newblock {\sl \bibinfo{journal}{Quantum Science and Technology}}
  \bibinfo{volume}{4}(\bibinfo{number}{1}), p. \bibinfo{pages}{015004},
  \doi{10.1088/2058-9565/aad604}.
\newblock \bibinfo{note}{Also available from \arxiv{1712.01557}}.

\bibitemdeclare{book}{ant-book}
\bibitem{ant-book}
\bibinfo{author}{Kenneth \surnamestart Ireland\surnameend} \&
  \bibinfo{author}{Michael \surnamestart Rosen\surnameend}
  (\bibinfo{year}{1982}): \emph{\bibinfo{title}{A Classical Introduction to
  Modern Number Theory}}.
\newblock \bibinfo{series}{Graduate Texts in Mathematics 84},
  \bibinfo{publisher}{Springer}, \doi{10.1007/978-1-4757-2103-4}.

\bibitemdeclare{article}{Kissinger-Wetering}
\bibitem{Kissinger-Wetering}
\bibinfo{author}{Aleks \surnamestart Kissinger\surnameend} \&
  \bibinfo{author}{John \surnamestart van~de Wetering\surnameend}
  (\bibinfo{year}{2020}): \emph{\bibinfo{title}{Reducing the number of
  non-{Clifford} gates in quantum circuits}}.
\newblock {\sl \bibinfo{journal}{Phys. Rev. A}} \bibinfo{volume}{102}, p.
  \bibinfo{pages}{022406}, \doi{10.1103/PhysRevA.102.022406}.

\bibitemdeclare{inproceedings}{Li-Ross-Selinger}
\bibitem{Li-Ross-Selinger}
\bibinfo{author}{Sarah~Meng \surnamestart Li\surnameend},
  \bibinfo{author}{Neil~J. \surnamestart Ross\surnameend} \&
  \bibinfo{author}{Peter \surnamestart Selinger\surnameend}
  (\bibinfo{year}{2021}): \emph{\bibinfo{title}{Generators and relations for
  the group {$O_n(\Z[1/2])$}}}.
\newblock In: {\sl \bibinfo{booktitle}{Proceedings of the 18th International
  Conference on Quantum Physics and Logic, QPL 2021}},
  \bibinfo{series}{Electronic Proceedings in Theoretical Computer Science}.
\newblock \bibinfo{note}{Also available from \arxiv{2106.01175}}.

\bibitemdeclare{article}{Nam_2018}
\bibitem{Nam_2018}
\bibinfo{author}{Yunseong \surnamestart Nam\surnameend},
  \bibinfo{author}{Neil~J. \surnamestart Ross\surnameend},
  \bibinfo{author}{Yuan \surnamestart Su\surnameend},
  \bibinfo{author}{Andrew~M. \surnamestart Childs\surnameend} \&
  \bibinfo{author}{Dmitri \surnamestart Maslov\surnameend}
  (\bibinfo{year}{2018}): \emph{\bibinfo{title}{Automated optimization of large
  quantum circuits with continuous parameters}}.
\newblock {\sl \bibinfo{journal}{Npj Quantum Information}}
  \bibinfo{volume}{4}(\bibinfo{number}{1}), \doi{10.1038/s41534-018-0072-4}.
\newblock \bibinfo{note}{Also available from \arxiv{1710.07345}}.

\bibitemdeclare{unpublished}{ZhangCheng19}
\bibitem{ZhangCheng19}
\bibinfo{author}{Fang \surnamestart {Zhang}\surnameend} \&
  \bibinfo{author}{Jianxin \surnamestart {Chen}\surnameend}
  (\bibinfo{year}{2019}): \emph{\bibinfo{title}{Optimizing {$T$} gates in
  {Clifford+$T$} circuit as $\pi/4$ rotations around {Paulis}}}.
\newblock \bibinfo{note}{{Available from \arxiv{1903.12456}}}.

\end{thebibliography}

% ----------------------------------------------------------------------
\newpage
\appendix
\section{Appendix: Proof of the Main Lemma}

\parskip=0.85em

Before we prove Lemma~\ref{lem:main}, we collect a number of useful
consequences of the relations from Figure~\ref{rels}. These are shown
in Figure~\ref{useful}. The proofs of these relations are
straightforward.

To keep the proof of Lemma~\ref{lem:main} as readable as possible, we
make the following simplification. Each time we complete the diagram
{\eqref{eqn:lem:main}}, we will permit $\vec G'$ to be a sequence of
simple edges, rather than basic edges as required by the lemma. This
is justified because each such sequence of simple edges can be
expanded into a (usually much longer) sequence of basic edges whose
level is no higher than that of the original sequence.

With that in mind, we now proceed to prove Lemma~\ref{lem:main} by
case distinction. Assume $s\xrightarrow{G}r$ is a basic edge of the
Cayley graph, and $s\xRightarrow{N}t$ is a normal edge. Let
$L=(p,k,m)$ be the level of $s$; specifically, $p$ is the index of the
pivot column, $k$ is the least denominator exponent of the pivot
column, $m$ is the number of odd entries in $\gamma^k v$, where $v$ is
the pivot column. Also let $w=\gamma^k v$.

% 1.
\case{$G = i_{[0]}$} Let $\jay$ be the index of the first odd entry of
$w$.
\subcase{$\jay > 0$}
Note that $0<j\leq p$. Then the normal edge $N$ does not act on row 0,
and $N$ is still the normal edge for state $r$. We complete the
diagram as follows.
\[\xymatrix{
  s
  \ar[r]^<>(.5){i_{[0]}}
  \ar@{=>}[d]_<>(.5){N} &
  r
  \ar@{=>}[d]^<>(.5){N} \\
  t
  \ar[r]^<>(.5){i_{[0]}}
  &
  q
}
\]
The diagram commutes relationally by {\eqref{eq:4}--\eqref{eq:9}},
which ensure that disjoint generators commute. We will encounter
additional cases in which the states $s$ and $r$ generate the same
syllable $N$ and the indices that $N$ acts on are disjoint from those
of $G$. We will refer to such cases as ``disjoint'' cases.

\subcase{$j=0$, $k=0$, and $p=0$} Since the $j$th entry of $v$ is odd,
by Lemma~\ref{lem:lde0}, it must be of the form $i^e$ for some
$e\in\s{0,\ldots,3}$. Note that $e=0$ is not possible, because then
$v$ would be $e_p$ and would not be a pivot column. Therefore
$e>0$. In that case, the normal edge from $s$ is $i_{[0]}^{4-e}$ and
we complete the diagram as follows.
\[\xymatrix{
  s
  \ar[r]^<>(.5){i_{[0]}}
  \ar@{=>}[d]_<>(.5){i_{[0]}^{4-e}} &
  r
  \ar@{=>}[d]^<>(.5){i_{[0]}^{3-e}} \\
  t
  \ar[r]_<>(.5){\eps} &
  t
}
\]
Note that here, $\xrightarrow{\eps}$ denotes a path of length
$0$. Also, $\xRightarrow{i_{[0]}^{3-e}}$ denotes a path of length $0$
if $e=3$ and a path of length 1 otherwise. The diagram commutes
relationally by reflexivity.

\subcase{$j=0$, $k=0$, and $p>0$} In this case, the exact synthesis
algorithm specifies that the normal edge from $s$ is
$\tlx{0}{p}i_{[0]}^{-e}$, and the normal edge from $r$ is
$\tlx{0}{p}i_{[0]}^{-e-1}$. Here and from now on, we will tacitly
understand all exponents of $i_{[j]}$ to be taken modulo 4, which is
justified by relation {\eqref{eq:1}}. Similarly, from now on we will
also tacitly use relations {\eqref{eq:2}} and {\eqref{eq:3}} to invert
the $X$ and $K$ generators when appropriate. We complete the diagram
as follows.
\[\xymatrix{
  s
  \ar[r]^<>(.5){i_{[0]}}
  \ar@{=>}[d]_<>(.5){\tlx{0}{p}i_{[0]}^{-e}} &
  r
  \ar@{=>}[d]^<>(.5){\tlx{0}{p}i_{[0]}^{-e-1}} \\
  t
  \ar[r]_<>(.5){\eps} &
  q
}
\]

% ......................................................................
\begin{figure}
  \begin{align}
    \tlkd{\jay}{\ell}\tli{\jay}{} & \,\simr\,\tli{\jay}{}\tli{\ell}{}\tlx{\jay}{\ell}\tlkd{\jay}{\ell}\tli{\ell}{} \label{eq:case1.4b} \\
    \tlk{j}{\ell} &\,\simr\, \tli{j}{3}\tli{\ell}{3}\tlkd{j}{\ell} \label{eq:case2.2.1a} \\
    \tlkd{j}{\ell}\tli{\ell}{}\tlk{j}{\ell} & \,\simr\, \tli{\ell}{3}\tlx{j}{\ell}\tlkd{j}{\ell}\tli{\ell}{} \label{eq:case2.2.1b} \\
    X_{[j,\ell]}i^{q}_{[j]}\tlx{\kay}{\ell} & \,\simr\, X_{[j,\kay]}X_{[j,\ell]}i^{q}_{[j]} \label{eq:case3.1.2.3} \\
    X_{[\kay,\ell]}i^{q}_{[\kay]}\tlx{\jay}{\kay} & \,\simr\, X_{[\jay,\kay]}X_{[\jay,\ell]}i^{q}_{[\jay]} \label{eq:case3.1.3.1} \\
    \tlk{j}{\ell}i^{q}_{[\ell]}\tlx{\kay}{\ell} & \,\simr\, \tlx{\kay}{\ell}\tlk{j}{\kay}i^{q}_{[\kay]}\label{eq:case3.2.2.2a} \\
    \tlkd{\ell}{\ell'}\tlkd{j}{j'}\tlkd{j'}{\ell'}\tlkd{j}{\ell}
    \tlx{\ell}{\jay'} & \,\simr\, \tlx{\ell}{\jay'}
    \tlkd{\ell}{\ell'}\tlkd{j}{j'}\tlkd{j'}{\ell'}\tlkd{j}{\ell}
    \label{eq:case3.2.2.2.1} \\
    \tlkd{\jay}{\ell}i^{}_{[\ell]} \tlx{\jay}{\ell} & \,\simr\,
    \tlx{\jay}{\ell}\tli{\jay}{3}\tli{\ell}{} \tlkd{\jay}{\ell}i^{}_{[\ell]}\label{eq:case3.2.2.3}
  \end{align}
  
  \caption{Some useful relations in $\simr$. All of these are
    consequences of the relations in Figure~\ref{rels}.  As before, in
    each relation, the indices are assumed to be distinct, and when a
    generator $\tlx{a}{b}$ or $\tlk{a}{b}$ is mentioned, we assume
    $a<b$. $K\da$ abbreviates $K^7$.}
  \label{useful}
  \hfill\rule{0.95\textwidth}{0.1mm}\hfill\hbox{}
\end{figure}
% ......................................................................

\subcase{$j=0$ and $k>0$} By Lemma~\ref{lem:evenodd}, $w$ has an even
number of odd entries. Let $\ell$ be the index of the second odd entry
of $w$. In this case, the exact synthesis algorithm specifies that the
normal edge from $s$ is $\tlkd{0}{\ell}\tli{\ell}{q}$ and the normal
edge from $r$ is $\tlkd{0}{\ell}\tli{\ell}{q'}$, for $q,q'\in\s{0,1}$
and $q'\neq q$.  We complete the diagram as follows:
\[\xymatrix@C+6ex{
  s
  \ar[r]^<>(.5){i_{[0]}}
  \ar@{=>}[d]_<>(.5){\tlkd{0}{\ell}\tli{\ell}{q}} &
  r
  \ar@{=>}[d]^<>(.5){\tlkd{0}{\ell}\tli{\ell}{q'}} \\
  t
  \ar[r]^<>(.5){i_{[0]}i_{[\ell]}\tlx{0}{\ell}^q}
  &
  q
}
\]
This diagram commutes relationally by {\eqref{eq:15}} when
$q=0$ and by {\eqref{eq:case1.4b}} when $q=1$.

% 2.
\case{$G = \tlk{0}{1}$}

\subcase{$k=0$} Let $\jay$ be the index of the first odd entry of $w$.

\subsubcase{$\jay < 2$} In this case, the exact synthesis
algorithm specifies that the normal edge from $r$ is $\tlkd{0}{1}$. We
complete the diagram as follows, and it commutes relationally by
{\eqref{eq:3}}.
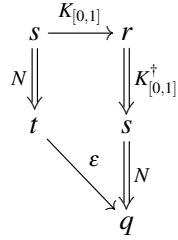
\[\xymatrix{
  s
  \ar[r]^<>(.5){\tlk{0}{1}}
  \ar@{=>}[d]_<>(.5){N} &
  r
  \ar@{=>}[d]^<>(.5){\tlkd{0}{1}} \\
  t
  \ar[rd]^<>(.5){\eps}
  &s
  \ar@{=>}[d]^<>(.5){N} \\
  &q
}
\]
We will encounter additional cases in which the normal edge from $r$
is relationally the inverse of $G$. In such cases, the diagram can
always be completed in the same way. We refer to these cases as
``retrograde''.

\subsubcase{$j\geq 2$} Note that $v_0=v_1=0$. Here, and from
now on, we write $v_j$ for the $j$th component of a vector $v$. Since
$j\leq p$, the normal edges from both $s$ and $r$ are
$X_{[j,p]}i^{-e}_{[j]}$, which are disjoint from $\oh_{[0,1]}$. This
is a disjoint case.

\subcase{$k>0$} By Lemma~\ref{lem:evenodd}, $w$ has an even number of
odd entries. Let $\jay$ and $\ell$ be the indices of the first two odd
entries of $w$.

\subsubcase{$\jay=0$ and $\ell=1$} In this case, the exact synthesis
algorithm specifies that the normal edge from $s$ is
$\tlkd{0}{1}\tli{1}{q}$. If $q=0$, then $\level{r}<\level{s}$, and we
complete the diagram as follows. It commutes relationally by
{\eqref{eq:case2.2.1a}}.
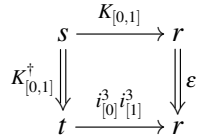
\[\xymatrix@C+2ex{
  s
  \ar[r]^<>(.5){\tlk{0}{1}}
  \ar@{=>}[d]_<>(.5){\tlkd{0}{1}} &
  r
  \ar@{=>}[d]^<>(.5){\eps} \\
  t
  \ar[r]^<>(.5){\tli{0}{3}\tli{1}{3}}
  &r
}
\]
If $q=1$, then the exact synthesis algorithm specifies that the normal
edge from $r$ is also $\tlkd{0}{1}\tli{1}{}$. In this case, we
complete the diagram as follows. It commutes relationally by
{\eqref{eq:case2.2.1b}}.
\[\xymatrix@C+2ex{
  s
  \ar[r]^<>(.5){\tlk{0}{1}}
  \ar@{=>}[d]_<>(.5){\tlkd{0}{1}\tli{1}{}} &
  r
  \ar@{=>}[d]^<>(.5){\tlkd{0}{1}\tli{1}{}} \\
  t
  \ar[r]^<>(.5){\tli{1}{3}\tlx{0}{1}}
  &r
}
\]

\subsubcase{$j=0$ and $\ell>1$} Note that $j<\ell\leq p$, so the first
two entries in each column after the pivot column are $0$, hence
$\tlk{0}{1}$ does not change $p$. But it will increase the least
denominator exponent of the pivot column from $k$ to $k+1$.  The exact
synthesis algorithm then specifies that the normal edge from $r$ is
$\tlkd{0}{1}$, so this case is retrograde.

\subsubcase{$j=1$} This is similar to the previous case. Again,
$\tlk{0}{1}$ increases the denominator exponent of the pivot column,
the normal edge is $\tlkd{0}{1}$, and so the case is retrograde.
      
\subsubcase{$j>1$} In this case, $\lde{v_0, v_1}<k$. Note that
$\jay<\ell\leq p$, so the first two entries in each column after the
pivot column are $0$, hence $\tlk{0}{1}$ does not change $p$. The
exact synthesis algorithm specifies that the normal edge from $s$ is
$\tlkd{\jay}{\ell}\tli{\ell}{q}$. Let $v'$ be the pivot column of
$r$.

If $\lde{v'_0, v'_1}<k$, then the normal edge from $r$ is also
$\tlkd{\jay}{\ell}\tli{\ell}{q}$ and the case is disjoint.  On the
other hand, if $\lde{v'_0, v'_1}=k$. Let $w'=\gamma^kv'$. One can show
that in this case, $w'_0\equiv w'_1\Mod{\gamma^2}$, and therefore the
normal edge from $r$ is $\tlkd{0}{1}$. Then this case is retrograde.

% 3.
\case{$G = X_{[\alpha,\alpha+1]}$}

\subcase{$k=0$} Let $\jay$ be the index of the first odd entry of $v$.
Note that $\jay\leq p$. Also, by Lemma~\ref{lem:lde0}, the $j$th entry
of $v$ is of the form $i^e$ for some $e\in\s{0,\ldots,3}$.

\subsubcase{$\alpha \geq p$} Applying $X_{[\alpha,\alpha+1]}$
increases $p$, and the exact synthesis algorithm specifies that the
normal edge from $r$ is $X_{[\alpha,\alpha+1]}$. Hence this case is
retrograde using {\eqref{eq:2}}.

\subsubcase{$\alpha = p - 1$} 

\subsubsubcase{$j = \alpha+1$} Note that $e=0$ is not possible,
because then $v$ would be $e_p$ and would not be a pivot
column. Therefore $e>0$. The exact synthesis algorithm specifies that
the normal edge from $s$ is $i_{[\alpha+1]}^{-e}$ and the normal edge
from $r$ is $\tlx{\alpha}{\alpha+1}\tli{\alpha}{-e}$.  We complete the
diagram as follows, and it commutes relationally by {\eqref{eq:10}}
and {\eqref{eq:2}}.

\[\xymatrix@C+4ex{
  s
  \ar[r]^<>(.5){\xal}
  \ar@{=>}[d]_<>(.5){i^{-e}_{[\alpha+1]}} &
  r
  \ar@{=>}[d]^<>(.5){\xal i^{-e}_{[\alpha]}} \\
  t
  \ar[r]^<>(.5){\eps} &
  q \\
}
\]

\subsubsubcase{$j = \alpha$} The exact synthesis algorithm specifies
$\tlx{\alpha}{\alpha+1}\tli{\alpha}{-e}$ and $\tli{\alpha+1}{-e}$ as
the normal edges from $s$ and $r$, respectively. We complete the
diagram as follows, and it commutes relationally by {\eqref{eq:10}}.
\[\xymatrix@C+4ex{
  s
  \ar[r]^<>(.5){\xal}
  \ar@{=>}[d]_<>(.5){\tlx{\alpha}{\alpha+1}\tli{\alpha}{-e}} &
  r
  \ar@{=>}[d]^<>(.5){\tli{\alpha+1}{-e}} \\
  t
  \ar[r]^<>(.5){\eps} &
  q \\
}
\]

\subsubsubcase{$j \leq \alpha-1$} The exact synthesis algorithm
specifies that the normal edge from both $s$ and $r$ is
$X_{[j,\alpha+1]}i^{-e}_{[j]}$. We complete the diagram as follows. It
commutes relationally by {\eqref{eq:case3.1.2.3}}.
\[\xymatrix@C+4ex{
  s
  \ar[r]^<>(.5){\xal}
  \ar@{=>}[d]_<>(.5){X_{[j,\alpha+1]}i^{-e}_{[j]}} &
  r
  \ar@{=>}[d]^<>(.5){X_{[j,\alpha+1]}i^{-e}_{[j]}} \\
  t
  \ar[r]^<>(.5){X_{[j,\alpha]}} &
  q \\
}
\]

\subsubcase{$\alpha \leq p-2$} 

\subsubsubcase{$j = \alpha$} The exact synthesis algorithm specifies
that the normal edge from $s$ is $X_{[\alpha,p]}i^{-e}_{[\alpha]}$ and
the normal edge from $r$ is $X_{[\alpha+1,p]}i^{-e}_{[\alpha+1]}$ We
complete the diagram as follows. It commutes relationally by
{\eqref{eq:case3.1.3.1}}.
\[\xymatrix@C+4ex{
  s
  \ar[r]^<>(.5){\xal}
  \ar@{=>}[d]_<>(.5){X_{[\alpha,p]}i^{-e}_{[\alpha]}} &
  r
  \ar@{=>}[d]^<>(.5){X_{[\alpha+1,p]}i^{-e}_{[\alpha + 1]}} \\
  t
  \ar[r]^<>(.5){X_{[\alpha,\alpha+1]}} &
  q \\
}
\]

\subsubsubcase{$j = \alpha+1$}
The exact synthesis algorithm specifies
that the normal edge from $s$ is $X_{[\alpha+1,p]}i^{-e}_{[\alpha+1]}$ and
the normal edge from $r$ is $X_{[\alpha,p]}i^{-e}_{[\alpha]}$ We
complete the diagram as follows. It commutes relationally by
{\eqref{eq:case3.1.3.1}} and {\eqref{eq:2}}.
\[\xymatrix@C+4ex{
  s
  \ar[r]^<>(.5){\xal}
  \ar@{=>}[d]_<>(.5){X_{[\alpha+1,p]}i^{-e}_{[\alpha+1]}} &
  r
  \ar@{=>}[d]^<>(.5){X_{[\alpha,p]}i^{-e}_{[\alpha ]}} \\
  t
  \ar[r]^<>(.5){X_{[\alpha,\alpha+1]}} &
  q \\
}
\]

\subsubsubcase{$j \neq \alpha$ and $j \neq \alpha+1$} In this case,
both normal edges are $X_{[\jay,p]}i^{-e}_{[\jay]}$. This case is
disjoint.

\subcase{$k>0$} By Lemma~\ref{lem:evenodd}, $w$ has an even number of
odd entries. Let $\jay$ and $\ell$ be the indices of the first two odd
entries of $w$.

\subsubcase{$ \alpha \geq p$} Applying $X_{[\alpha,\alpha+1]}$
increases $p$, and the normal edge from $r$ is $X_{[\alpha,\alpha+1]}$. Therefore this case is retrograde.

\subsubcase{$\alpha < p$} We will do a case distinction on how
$\jay<\ell$ overlaps with $\alpha<\alpha+1$.

\subsubsubcase{$\ell < \alpha$} The normal edge from both $s$ and $r$
is $\tlkd{\jay}{\ell}\tli{\ell}{q}$, so this case is disjoint.

\subsubsubcase{$\ell = \alpha$} The exact synthesis algorithm
specifies that the normal edge from $s$ is
$\tlkd{\jay}{\alpha}\tli{\alpha}{q}$, for some $q\in\s{0,1}$.

If $w_{\alpha+1}$ is even, then the normal edge from $r$ will be
$\tlkd{\jay}{\alpha+1}\tli{\alpha+1}{q}$. In this case, we can
complete the diagram as follows. It commutes relationally by
{\eqref{eq:case3.2.2.2a}}.
\[\xymatrix{
  s
  \ar[rr]^<>(.5){\xal}
  \ar@{=>}[d]_<>(.5){\tlk{j}{\alpha}i^{q}_{[\alpha]}} &&
  r
  \ar@{=>}[d]^<>(.5){\tlk{j}{\alpha+1}i^{q}_{[\alpha+1]}} \\
  t
  \ar[rr]^<>(.5){\xal} &&
  q \\
}
\]
Now assume that $w_{\alpha+1}$ is odd. By Lemma~\ref{lem:evenodd}, we
have a fourth odd entry. Let $\jay'= \alpha+1$ and $\ell'$ be the
index of the fourth odd entry. By Lemma~\ref{lem:modulo-1-3}, we have
$w_{\jay} = i^e + a\gamma^3$, $w_{\ell} = i^f + b\gamma^3$, $w_{\jay'}
= i^g + c\gamma^3$, and $w_{\ell'} = i^h + d\gamma^3$, for some
$e,f,g,h\in\s{0,\ldots,3}$ and $a,b,c,d\in\Z[i]$.

\subsubsubsubcase{$e=f=g=h=0$} In this special case, the normal edges
from $s$ and $r$ are both $\tlkd{\jay}{\ell}$. We complete the diagram
as follows.

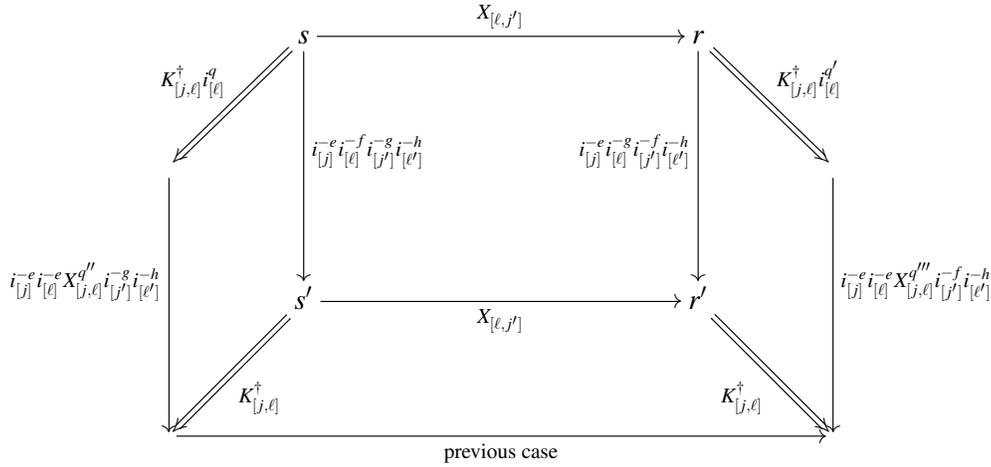
\[
\xymatrix@C+4ex{
  s
  \ar[rrrrrrr]^<>(.5){\tlx{\ell}{\jay'}}
  \ar@{=>}[d]_<>(.5){\tlkd{j}{\ell}} &&&&&&&
  r
  \ar@{=>}[d]^<>(.5){\tlkd{j}{\ell}} \\
  t
  \ar[r]^<>(.5){\tlkd{j'}{\ell'}} &
  \ar[r]^<>(.5){\tlkd{j}{j'}} &
  \ar[r]^<>(.5){\tlkd{\ell}{\ell'}} &
  \ar[r]^<>(.5){\tlx{\ell}{\jay'}} &
  \ar[r]^<>(.5){\tlk{\ell}{\ell'}} &
  \ar[r]^<>(.5){\tlk{j}{j'}} &
  \ar[r]^<>(.5){\tlk{j'}{\ell'}} &
  q \\
}
\]
The fact that this diagram commutes relationally follows from
{\eqref{eq:case3.2.2.2.1}}. We must verify that it satisfies the
level condition. The following diagram shows only entries
$j,\ell,j',\ell'$ of $\gamma^k$ times the $p$th column of each state.
Since in all 8 states below the top row, at least two entries are
even, the level of all of these states is strictly less than that of
$s$.
\[\xymatrix{
  *{\begin{pmatrix}
      1 + a \gamma^3\\
      1 + b \gamma^3\\
      1 + c \gamma^3\\
      1 + d \gamma^3\\
  \end{pmatrix}}
  \ar[rr]^<>(.5){\tlx{\ell}{\jay'}}
  \ar@{=>}[d]_<>(.5){\tlkd{j}{\ell}} &&
  *{\begin{pmatrix}
      1 + a \gamma^3\\
      1 + c \gamma^3\\
      1 + b \gamma^3\\
      1 + d \gamma^3\\
  \end{pmatrix}}
  \ar@{=>}[d]^<>(.5){\tlkd{j}{\ell}}
  \\
  *{\begin{pmatrix}
      (1+i) + i(a+b) \gamma^2\\
      i(a-b) \gamma^2\\
      1 + c \gamma^3\\
      1 + d \gamma^3\\
  \end{pmatrix}}
  \ar@{->}[d]_<>(.5){\tlkd{j'}{\ell'}}&&
  *{\begin{pmatrix}
      (1+i) + i(a+c) \gamma^2\\
      i(a-c) \gamma^2\\
      1 + b \gamma^3\\
      1 + d \gamma^3\\
  \end{pmatrix}}
  \ar@{<-}[d]^<>(.5){\tlk{j'}{\ell'}}
  \\
  *{\begin{pmatrix}
      (1+i) + i(a+b) \gamma^2\\
      i(a-b) \gamma^2\\
      (1+i) + i(c+d) \gamma^2\\
      i(c-d) \gamma^2\\
  \end{pmatrix}}
  \ar@{->}[d]_<>(.5){\tlkd{j}{j'}}&&
  *{\begin{pmatrix}
      (1+i) + i(a+c) \gamma^2\\
      i(a-c) \gamma^2\\
      (1+i) + i(b+d) \gamma^2\\
      i(b-d) \gamma^2\\
  \end{pmatrix}}
  \ar@{<-}[d]^<>(.5){\tlk{j}{j'}}
  \\
  *{\begin{pmatrix}
      2i - (a+b+c+d)\gamma \\
      i(a-b) \gamma^2\\
      - (a+b-c-d)\gamma \\
      i(c-d) \gamma^2\\
  \end{pmatrix}}
  \ar@{->}[d]_<>(.5){\tlkd{\ell}{\ell'}}&&
  *{\begin{pmatrix}
      2i - (a+c+b+d)\gamma \\
      i(a-c) \gamma^2\\
      - (a+c-b-d)\gamma \\
      i(b-d) \gamma^2\\
  \end{pmatrix}}
  \ar@{<-}[d]^<>(.5){\tlk{\ell}{\ell'}}
  \\
  *{\begin{pmatrix}
      2i - (a+b+c+d)\gamma \\
      -(a-b+c-d)\gamma\\
      - (a+b-c-d)\gamma \\
      -(a-b-c+d)\gamma\\
  \end{pmatrix}}
  \ar[rr]^<>(.5){\tlx{\ell}{\jay'}}&&
  *{\begin{pmatrix}
      2i - (a+c+b+d)\gamma \\
      -(a-c+b-d)\gamma\\
      - (a+c-b-d)\gamma \\
      -(a-c-b+d)\gamma\\
  \end{pmatrix}}
}
\]

\subsubsubsubcase{Otherwise} In the general case, the normal edge from
$s$ is $\tlkd{\jay}{\ell}i^{q}_{[\ell]}$, where $q=0$ if $e-f$ is even
and $q=1$ if $e-f$ is odd. Similarly, the normal edge from $r$ is
$\tlkd{\jay}{\ell}i^{q'}_{[\ell]}$, where $q'=0$ if $e-g$ is even and
$q'=1$ if $e-g$ is odd.  Define $q'' = 0$ when $-q-f+e\equiv 0\Mod{4}$
and $q''=1$ when $e-f-q\equiv 2\Mod{4}$. Similarly, define $q''' = 0$
when $e-g-q'\equiv 0\Mod{4}$ and $q'''=1$ when $-q-g+e\equiv
2\Mod{4}$.  We complete the diagram using the outer perimeter of the
following figure. Here, the bottom edge is given as in the previous
case. 

\[\xymatrix@C+4ex@R+4ex{
  &s
  \ar[rrr]^<>(.5){\tlx{\ell}{\jay'}}
  \ar@{->}[dd]^<(0.4){\tli{\jay}{-e}\tli{\ell}{-f}\tli{\jay'}{-g}\tli{\ell'}{-h}}
  \ar@{=>}[ld]_<>(.5){\tlkd{\jay}{\ell}i^{q}_{[\ell]}}&&&
  r
  \ar@{=>}[rd]^<>(.5){\tlkd{\jay}{\ell}i^{q'}_{[\ell]}}
  \ar@{->}[dd]_<(0.4){\tli{\jay}{-e}\tli{\ell}{-g}\tli{\jay'}{-f}\tli{\ell'}{-h}}&\\
  \ar@{->}[dd]_<(0.4){\tli{\jay}{-e}\tli{\ell}{-e}\tlx{\jay}{\ell}^{q''}\tli{\jay'}{-g}\tli{\ell'}{-h}}
  &&&&&\ar@{->}[dd]^<(0.4){\tli{\jay}{-e}\tli{\ell}{-e}\tlx{\jay}{\ell}^{q'''}\tli{\jay'}{-f}\tli{\ell'}{-h}}\\
  & s'\ar@{=>}[ld]^<>(.5){\tlkd{\jay}{\ell}} \ar[rrr]_<>(.5){\tlx{\ell}{\jay'}}&&&r'\ar@{=>}[rd]_<>(.5){\tlkd{\jay}{\ell}}&\\
  \ar[rrrrr]_<>(.5){\text{previous case}}&&&&&\\          
}
\]

To see that the outer perimeter relationally commutes, it suffices to
show that the four inner faces relationally commutes. The bottom face
does so by the previous case. The middle face commutes by repeated
applications of {\eqref{eq:5}} and {\eqref{eq:10}}. To see why the
left face commutes, we first note that
\begin{equation}\label{eq:q-prime-prime}
  \tlkd{\jay}{\ell}\,\tli{\ell}{e-f-q}
  \,\simr\,
  \tlx{\jay}{\ell}^{q''}\,\tlkd{\jay}{\ell}.
\end{equation}
Namely, this holds by {\eqref{eq:13}} in case $e-f-q\equiv 2\Mod{4}$
and $q''=1$, and it holds trivially in case $e-f-q\equiv 0\Mod{4}$ and
$q''=0$. Then the left face commutes because
\begin{align*}
  \tlkd{\jay}{\ell}\,\tli{\jay}{-e}\,\tli{\ell}{-f}\,\tli{\jay'}{-g}\,\tli{\ell'}{-h}
  &\,\simr\,
  \tlkd{\jay}{\ell}\,\tli{\jay}{-e}\,\tli{\ell}{-e}\,\tli{\ell}{e-f-q}\,\tli{\ell}{q}\tli{\jay'}{-g}\,\tli{\ell'}{-h}
  && \mbox{by {\eqref{eq:1}}}
  \\
  &\,\simr\,
  \tli{\jay}{-e}\,\tli{\ell}{-e}\,\tlkd{\jay}{\ell}\,\tli{\ell}{e-f-q}\,\tli{\ell}{q}\tli{\jay'}{-g}\,\tli{\ell'}{-h}
  && \mbox{by {\eqref{eq:15}}}
  \\
  &\,\simr\,
  \tli{\jay}{-e}\,\tli{\ell}{-e}\,\tlx{\jay}{\ell}^{q''}\,\tlkd{\jay}{\ell}\,i^{q}_{[\ell]}\,\tli{\jay'}{-g}\,\tli{\ell'}{-h}
  &&\mbox{by {\eqref{eq:q-prime-prime}}}
  \\
  &\,\simr\,
  \tli{\jay}{-e}\,\tli{\ell}{-e}\,\tlx{\jay}{\ell}^{q''}\,\tli{\jay'}{-g}\,\tli{\ell'}{-h}\,\tlkd{\jay}{\ell}\,i^{q}_{[\ell]}
  &&\mbox{by {\eqref{eq:4}} and {\eqref{eq:6}}}.
\end{align*}
The right face commutes for the same reason, just swapping the rules
of $f$ and $g$. Finally, we need to verify that the diagram satisfies
the level condition. To this end, note that $s$, $r$, $s'$, and $r'$
all have the same level, because the operations $\tlx{\ell}{\jay'}$
and $\tli{\jay}{-e}\,\tli{\ell}{-f}\,\tli{\jay'}{-g}\,\tli{\ell'}{-h}$
neither change the pivot column, the denominator exponent, nor the
number of odd entries. Together with the fact that normal edges are
level decreasing and with what was shown in the previous case, this
implies the level condition.

\subsubsubcase{$\ell = \alpha+1$ and $\jay=\alpha$} In this case, the
exact synthesis algorithm prescribes that the normal edges from both
$s$ and $r$ are $\tlkd{\alpha}{\alpha+1}\tli{\alpha+1}{q}$, for
$q\in\s{0,1}$. We complete the diagram as follows. It commutes
relationally by the inverse of {\eqref{eq:13}} when $q=0$ and
by {\eqref{eq:case3.2.2.3}} when $q=1$.

\[\xymatrix@C+4ex{
  s
  \ar[rr]^<>(.5){\xal}
  \ar@{=>}[d]_<>(.5){\tlkd{\alpha}{\alpha+1}i^{q}_{[\alpha+1]}} &&
  r
  \ar@{=>}[d]^<>(.5){\tlkd{\alpha}{\alpha+1}i^{q}_{[\alpha+1]}} \\
  t
  \ar[rr]_<>(.5){\tlx{\alpha}{\alpha+1}^q\tli{\alpha}{-q}\tli{\alpha+1}{2-q}} &&
  q \\
}
\]

\subsubsubcase{$\ell = \alpha+1$ and $\jay \neq \alpha$} In this case,
the exact synthesis algorithm specifies that the normal edge from $s$
is $\tlkd{\jay}{\alpha+1}\tli{\alpha+1}{q}$ and the normal edge from
$r$ is $\tlkd{\jay}{\alpha}\tli{\alpha}{q}$. We complete the diagram as
follows. It commutes relationally by {\eqref{eq:10}} and {\eqref{eq:12}}.
\[\xymatrix{
  s
  \ar[rr]^<>(.5){\xal}
  \ar@{=>}[d]_<>(.5){\tlkd{\jay}{\alpha+1}i^{q}_{[\alpha+1]}} &&
  r
  \ar@{=>}[d]^<>(.5){\tlkd{\jay}{\alpha}i^{q}_{[\alpha]}} \\
  t
  \ar[rr]_<>(.5){\xal} &&
  q \\
}
\]
\subsubsubcase{$\ell > \alpha + 1$ and $\jay < \alpha$}
This case is disjoint.

\subsubsubcase{$\ell > \alpha + 1$ and $\jay = \alpha$} In this case,
the exact synthesis algorithm specifies that the normal edge from $s$
is $\tlkd{\alpha}{\ell}\tli{\ell}{q}$ and the normal edge from
$r$ is $\tlkd{\alpha+1}{\ell}\tli{\ell}{q}$. We complete the diagram
as follows. It commutes relationally by {\eqref{eq:5}} and
{\eqref{eq:12}}.
\[\xymatrix{
  s
  \ar[rr]^<>(.5){\xal}
  \ar@{=>}[d]_<>(.5){\tlkd{\alpha}{\ell}i^{q}_{[\ell]}} &&
  r
  \ar@{=>}[d]^<>(.5){\tlkd{\alpha+1}{\ell}i^{q}_{[\ell]}} \\
  t
  \ar[rr]_<>(.5){\xal} &&
  q \\
}
\]

\subsubsubcase{$\ell > \alpha + 1$ and $\jay = \alpha+1$} In this
case, the exact synthesis algorithm specifies that the normal edge
from $s$ is $\tlkd{\alpha+1}{\ell}\tli{\ell}{q}$ and the normal edge
from $r$ is $\tlkd{\alpha}{\ell}\tli{\ell}{q}$. We complete the
diagram as follows. It commutes relationally by {\eqref{eq:5}} and
{\eqref{eq:12}}.

\[\xymatrix{
  s
  \ar[rr]^<>(.5){\xal}
  \ar@{=>}[d]_<>(.5){\tlkd{\alpha+1}{\ell}i^{q}_{[\ell]}} &&
  r
  \ar@{=>}[d]^<>(.5){\tlkd{\alpha}{\ell}i^{q}_{[\ell]}} \\
  t
  \ar[rr]_<>(.5){\xal} &&
  q \\
}
\]

\subsubsubcase{$\ell > \alpha + 1$ and $\jay > \alpha+1$}
This case is disjoint.

This finishes the proof of Lemma~\ref{lem:main}, and therefore of
completeness.

% ----------------------------------------------------------------------

\end{document}